\documentclass[11pt]{article}
\usepackage{amsmath,amssymb,amsthm,color,graphicx, fancybox}

\usepackage{subfigure}% subcaptions for subfigures
\usepackage{subfigmat}% matrices of similar subfigures, aka small mulitples
\usepackage{float}

\usepackage{amsmath}
\usepackage{tcolorbox}
\usepackage{bclogo}
\usepackage[latin1]{inputenc}
\usepackage{mdframed}
\usepackage{amsmath}
\usepackage{amsfonts}
\usepackage{amssymb}
\usepackage{wasysym}
\usepackage{amsthm}
\usepackage{graphicx}

\usepackage{subfigure}% subcaptions for subfigures
\usepackage{subfigmat}% matrices of similar subfigures, aka small mulitples
\usepackage{float}

\usepackage{subfigure}% subcaptions for subfigures
\usepackage{subfigmat}% matrices of similar subfigures, aka small mulitples
\usepackage{float}

\newcommand{\sa}{\lambda^{1/2}}

\newcommand{\R}{\mathbb R}

\newcommand{\eps}{\epsilon}

\newtheorem{theorem}{Theorem}[section]
\newtheorem{lemma}{Lemma}[section]
\newtheorem{prop}{Proposition}[section]
\newtheorem{cor}{Corollary}[section]

\newtheorem{remark}{Remark}[section]

\newtheorem{acknowledgment*}{Acknowledgment}

\newcommand{\be}{\begin{equation}}
\newcommand{\ee}{\end{equation}}
\newcommand{\bem}{\begin{multline}}
\newcommand{\eem}{\end{multline}}

\begin{document}
		\title{The linearized  Poisson-Nernst-Planck system as  heat flow on the interval under non-local boundary conditions }
	\author{Gershon Wolansky \\ Department of Mathematics, Technion, Haifa 32000, Israel}

	\maketitle
	\begin{abstract}
		The linearized  of the Poisson-Nernst-Planck (PNP)  equation under closed ends around a neutral state is studied. It is reduced to a damped heat equation under non-local boundary conditions, which leads to a stochastic interpretation of the linearized equation as a Brownian particle which jump and is reflected, at Poisson distributed time, to one of the end points of the channel, with a probability which is proportional to its distance from this end point.  An explicit expansion of the heat kernel reveals the eigenvalues and eigenstates of both the PNP equation and its adjoint. For this, we take advantage of the representation of the resulvent operator and recover the heat kernel by applying the  inverse Laplace transform.  
	\end{abstract}
	\section{Introduction}
The Poisson -Nernst-Planck (PNP) system  \cite{debay} is a fundamental model for electrodiffusion and is one of the main tools in modeling ion channels in cell membrane(see, e.g. \cite{mem}). In one of its simplest forms, it contains a pair of drift-diffusion equations for positive and negatively charged ions, coupled with the equation for the electric  field induced by the charges. 
	
	We concentrate  on the case of two  types of ions (positively ($C_+$) and negatively ($C_-$) charged) and  a closed channel, where the flux of $C_\pm$  is zero at the ends of the channel, hence the number of ions of each type  (and, in particular, the total charge $C_+-C_-$) is preserved in time. 
	
	The physical model behind the PNP is a drift diffusion for charged particles, where the diffusion is due to independent Brownian motions of the ions, and the drift is due to the external field induces by the potential difference between the ends of the channel, and the mean electric field generated by the moving ions. 
	
	Thus, the PNP can be considered as  a  system of Kolmogorov forward equation,
	\cite{kol, BM}, whose solutions represent the probability distribution of a test particle for each type of ions in the system. 

%In the case $D_+=D_-:= D$, 	if we linearize this system around the neutral state $C_+=\eta+c_+$, $C_-=\eta+c_-$ where $\eta>0$ is a constant, $E,V, c_\pm <<1$ we get (c.f. Appendix) a linear equation for 

In the case of zero external field,  the neutral case  ($C_+=C_-$) induces   a steady,  uniform distribution for both  charges. 
A linearization of this equation around this constant neutral case is reduced, up to a re-scaling of the time, into the naive looking damped diffusion equation \cite{pnp, mem}\footnote{ I wish to thank Dr. Doron Elad for turning my attention to this formulation} for the {\em local charge} $u\approx C_+-C_-$: 
%	$u=ze(C_+-C_-)$, which looks like a canonical, damped  heat  equation \cite{pnp}
	\be\label{dhe} u_t=u_{xx} -\kappa^2u \ , \  0<x<1, \ \ t\geq 0 \ee
	where the interval $[0,1]$ is the channel,  $u(x,t)$ is the local charge  at $x\in[0,1], t\geq 0$  and $\kappa^2$  is the inverse Debye screening length.   This looks like a fairly naive equation. However, the boundary conditions are 
	\be\label{bcc1}(u_x+(k^2/\eps)E)_{x=0,1}=0\ee
	where the electric field $E$ is given by the Poisson equation 
	\be\label{bcc2}-\eps E_x=u, \ \ \ \ \int_0^1 Edx=V \ ,\ee
	driven by the voltage difference $V$ across the end points $x=0,1$. These are non-local boundary conditions. Indeed, 
we show that (\ref{bcc1}, \ref{bcc2}) can be reduced to the following
$$    u_x(0) = -k^2\int_0^1(1-s) u(s)ds +\kappa^2V  \  , 
  u_x(1) = \kappa^2\int_0^1 s u(s)ds +\kappa^2V  \ . $$
  The steady state for the linearized problem can easily  be obtained:
  \be\label{baru}\bar{u} (x)= \frac{\kappa V}{\cosh(\kappa/2)}\sinh(\kappa( x-1/2))\ee
so we can subtract it from the solution $u$ of equation to get a homogeneous boundary conditions
\be\label{bcc3}    u_x(0) = -k^2\int_0^1(1-s) u(s)ds   \  , 
u_x(1) = \kappa^2\int_0^1 s u(s)ds  \ . \ee
Some versions of parabolic equations  under non-local boundary conditions were studied by several authors (see, e.g.\cite{St, Ardent}). 
% A simple integration implies that the integral of the solution is preserved.  In addition, the solutions also preserve sign (See Proposition \ref{prop1}).   Under the condition $u(x,0)\geq 0$, $\int_0^1 u(x,0)=1$ one can think about the solution as the probability density of canonical  diffusion equation. 
%Such initial data is not so natural by the interpretation of $u$ as the  local charge, which does not, necessarily, has a definite sign, but we can treat the negative and positive parts of the solution independently, using its linearity and the preservation of positive solutions. 
A stochastic  interpretation of linear  diffusion equations under non-local boundary conditions goes back to Feller \cite{Feller2}. In that paper he extended  his seminal paper \cite{Feller1} 
to non-local boundary conditions, and interpreted the diffusion equation in terms of a Brownian particle  which may undergo a jump from a point on the boundary of the interval to a distributed position at  the interior. This extension was later studied by several authors, see e.g \cite{Ardent, Bo, ros1}. However, in all these cases the process is allowed to jump from a boundary point to the interior, and not the other way around. This will be the case if, e.g., $\kappa^2$ is replaced by $-\kappa^2$ in (\ref{bcc3}).

The boundary conditions (\ref{bcc3}) associated with the operator $d^2/dx^2 -\kappa^2$ suggests a diffusion process which jump at a random Poisson time of mean $\kappa^2$ from an {\em 
inner}  point $x\in (0,1)$ and reflected at the endpoint $x=0$ with probability $1-x$, and at  the endpoint $x=1$ with probability $x$. In this sense, it is a {\em forward} Kolmogorov equation representing the evolution of a probability  distribution of the charge. \footnote{Even though $u$ is not necessarily  of definite sign, we can consider the positive and negative parts of $u$ independently, using the linearity of this equation.}

The {\em heat kernel} $K(x,y,t)$ of such an equation generates the solutions 
$$u(x,t)=\int_0^1 K(x,y,t)u(y,0)dy \ . $$
This kernel represents the probability of the particle to be at position $x$ at time $t+s$, conditioned that it was at point $y$ at time $s$. 
In particular
$$ K(x,y, t)\geq 0,\ \ \text{for} \ \ (x,y)\in (0,1)\times (0,1), \ t>0, \ \ \ \int_0^1K(x,y,t)dx=1$$
and $\lim_{t\downarrow 0} K(x,y,t)=\delta_{x-y}$. 
%Note also that the homogeneous system (\ref{dhe},\ref{bcc3}) admits a steady, positive ground state given by $\cosh(\kappa(x-1/2))$ (see Theorem \ref{th1}). 

The adjoint equation, then, represents the backward Kolmogorov equation  of the process. In general, it is also a diffusion equation of the same form and adapted boundary condition, whose kernel $K^{*}$ is given by the interchange of $x$ and $y$: $K^{*}(x,y,t)=K(y,x,t)$. However, in the case of b.c (\ref{bcc3}), an explicit form of the adjoint operator is not clear. %It can be obtain indirectly from the heat kernel $K(t,x,y)$ by exchanging $x,y$. 

% of the form
%$$ u(0)=\int_0^1 u(x)d\mu_0(x), \ \ \ u(1)=\int_0^1u(x)d\mu_1(x) $$
%where $\mu_0, \mu_1$ are {\em positive} measures. 
	In this paper we attempt to calculate the spectral expansion of the  heat kernel. The information encoded in this expansion contains the eigenvalues, as well as the eigenfunctions of both the operator and its adjoint.

	To obtain this, we take advantage on the explicit solutions of the {\em resulvent} $R=R(\lambda,x,y)$ where $\lambda\in \mathbb{C}$, $(x,y)\in (0,1)\times(0,1)$:
	$$ \partial^2_xR +\lambda R +\delta_{x-y}=0$$
where  $\delta_{x-y}$ is the Dirac delta function  and  $R$ satisfies  the boundary conditions (\ref{bcc3}) in the $x$ variable.  These solutions can be expressed locally as a combination of the trigonometric functions $\sin(\sa x)$ and $\cos(\sa x)$. It turns out that the solution of the resulvent equation exists whenever $Re(\lambda)<-\kappa^2$. 
	This  resulvent $R$ can also  be written in terms  of the heat kernel  $K$ 
(see \cite{DS}, and also the review in the Appendix):
	\be\label{Res} R(\lambda,x,y)= \int_0^\infty e^{(\lambda+\kappa^2)t}K(x,y,t)dt \ . \ee
	It turns out that $R$ is a meromorphic function of  $\lambda$ in the complex plane, analytic if $Re(\lambda)<-\kappa^2$, and admits a countable number of simple poles in the half plane $Re(\lambda\geq -\kappa^2)$ (including $\lambda=-\kappa^2$). %We now pose the assumption:
	
%	{\bf Assumption:} {\it All poles of the resulvent operator are on the real line}.
%	\vskip .2in Such initial data is not so natural by the interpretation of $u$ as the difference between positively and negatively charge ions (see section \ref{exap1}), but we can treat the negative and positive parts of the solution independently, using its linearity and the preservation of positive solutions (c.f Proposition \ref{prop1}).

	Under some conditions which we can verify (c.f. Appendix) we can recover the heat kernel form (\ref{Res}) using the inverse Laplace transform via
	$$ K(x,y,t)=e^{-\kappa^2 t}\frac{1}{2\pi i}\oint_\Gamma e^{-\lambda t}R(\lambda,x,y)d\lambda$$
	where $\Gamma$ is a contour enclosing all the poles of $R$. Then, we use the Residue theorem \cite{BS}  to evaluate the contour  integral.

	The main results are summarized below:
	\begin{theorem}\label{th1}
	%	Assume all eigenvalues associated with  the operator \\ $d^2/dx^2-\kappa^2$ under boundary conditions (\ref{bcc3}) are real
		%
	%	Then
		The heat kernel of (\ref{dhe}, \ref{bcc3}) is given by
	\begin{multline} \label{hetk}K(x,y,t)=\frac{\kappa\cosh(\kappa/2)\cosh(\kappa(x-1/2))}{2\sinh(\kappa/2)}\\+\sum_{n=1}^\infty
	\frac{\kappa^2 A(\lambda _n, y) (\sqrt{\lambda_n}+\kappa^2/\sqrt{\lambda_n})}{2\sqrt{\lambda_n}Det^{'}(\lambda_n)}
	\sin\left(\frac{\sqrt{\lambda_n}(1-2x)}{2}\right) e^{-(\lambda_n+\kappa^2) t} \\
	+ \frac{1}{2}\sum_{n=1}^\infty  \cos(2n\pi x)
	\left( \cos(2n\pi y)+\frac{1}{ n\pi(4n^2\pi^2+\kappa^2)} 
	\right)
	e^{-(4n^2\pi^2 +\kappa^2)t}\end{multline}
	where $\lambda _n$ are the roots
	%\footnote{It is conjectured that all eigenvalues $\sqrt{\lambda_k}^2$ are real, but  we cannot prove it at this stage...} 
	of 
		$$2\tan(\lambda_n^{1/2}/2)=\lambda_n^{1/2}\kappa^{-2}(\lambda_n+\kappa^2) \ , $$
		 $Det$ is given by (\ref{denum}) and    
	$$ A(\lambda, y):= \frac{1}{\pi^2} \sum_{m=1}^\infty \frac{2\cos((2m+1)\pi y)}{(2m+1)^2(\pi^2(2m+1)^2-\lambda)} \ . $$
	      
	\end{theorem}
	
	In particular,
	\begin{multline} \label{calL0}K_0(x,y,t):=\frac{\kappa\cosh(\kappa/2)}{2\sinh(\kappa/2)}\psi_0(x)\phi_0(y) e^{\kappa^2 t}+\sum_{n=1}^\infty
	\frac{\kappa^2 (\sqrt{\lambda_n}+\kappa^2/\sqrt{\lambda_n})}{2\sqrt{\lambda_n}Det^{'}(\lambda_n)}
	\psi^{(1)}_n(x)\phi^{(1)}(y)e^{-\lambda_n t} \\
	+ \frac{1}{2}\sum_{n=1}^\infty \psi_n^{(2)}(x)\phi_n^{(2)}(y)
	e^{-4n^2\pi^2 t}
	\end{multline}
	is the heat kernel for the operator ${\cal L}_0=d^2/dx^2$ on the domain  (\ref{bcc3} ), where 
\begin{itemize}
	\item 
$\psi^{(1)}_n(x)=\cos(2n\pi x)$, 
	\item $\psi^{(2)}_n(x)=\sin(\lambda_n^{1/2} (x-1/2))$
\item  $\psi_0(x)=\cosh(\kappa(x-1/2))$ 
\end{itemize}
 are the eigenstates of ${\cal L}_0$, and \begin{itemize}
 	\item 
 	 $\phi^{(1)}_n(y)=\cos(2n\pi y)+\frac{1}{ n\pi(4n^2\pi^2+\kappa^2)} $,
\item $\phi^{(2)}_n(y)=A(\lambda_n,y)$ ,
\item $\phi^{(0)}(y)=1$ 
\end{itemize}
 are the eigenstates of its adjoint ${\cal L}_0^\dag$. 
This poses a challenging question regarding the formulation of this problem, since (except of the constant), $\phi_n^{(1)}$, $\phi_n^{(2)}$ are not solutions of $\phi_{xx}+\lambda\phi=0$ for any $\lambda\in \mathbb{C}$. It seems  that the adjoint operator may not be given by a differential one, 
	{\em and       the non locality of the boundary conditions leaks into the operator itself. }(c.f Section \ref{conc}). 
	\begin{remark}
	All eigenvalues $\lambda_k$ for $k>0$ are real and positive (excluding the "ground" eigenvalue $\lambda_0=-\kappa^2$. 
	% This is supported by numerical evidence together with 
	\end{remark}

\section{The linearized PNP system}\label{exap1} 

	The one dimensional PNP equation takes the form \cite{pnp}
$$ C_{+,t}=D_+\left[ C_{+,x} + \frac{ze}{k_B T} EC_+ \right]_x $$
$$ C_{-,t}=D_-\left[ C_{-,x} - \frac{ze}{k_B T} EC_- \right]_x $$
on the interval $[0,1]$, 
where $C_\pm$ is the concentration of positive/negative ions,  and $E$ is the electric field given in terms of the concentrations $C_\pm$ and the potential difference $V$: 
$$ -\eps E _x = ze(C_+-C_-), \ \ \int_0^1 E(x,t)dx=V $$
The  special case of non penetrating charges corresponds to zero flux  on the boundary 
$$ C_{+,x}(0,t)+\frac{ze}{k_B T} C_+(0,t)E(0,t)= C_{-,x}(1,t)-\frac{ze}{k_B T} C_-(1,t)E(1,t)=0\ . $$

In the neutral case $C_+=C_-=\eta$ and $E=0$. We linearize this system :
$$C_+=\eta+c_+, \ \ C_-=\eta+c_-, \ \ \ E<<1$$
and ignore all terms of second order in $Ec_\pm$ 
to obtain
$$ u_t=D(u_{xx} -\kappa^2 u) + Bc_{xx}$$
$$c_{t}=B(c_{xx}-\kappa^2u)+u_{xx}$$
$$ -\eps E_x=u$$
where $D=\frac{D_++D_-}{2}$, $B=\frac{D_+-D_-}{2}$,  $\kappa^2=\frac{2\eta z e}{k_BT}$, $u=ze(c_+-c_-)$,\\ $c=ze(c_++c_-)$, subject to \\
 $\int_0^1Edx=V$, $(u_x+(\kappa^2/e)E)_{x=0,1}=0$, \ $(c_x)_{x=0,1}=0$. 
 
 Here we concentrate in the case $B=0$ which reduces to a single equation on $u$. Without loss of generality we also assume $D=1$:
 \be\label{pnp} -\eps E_x=u \ ,  \ \ \  \ \int_0^1Edx=V, \ \ (u_x+(\kappa^2\eps) E)_{x=0,1}=0\ . \ee
 %By scaling the time scale we can set $A=1$:
  \be\label{dheat} u_t=u_{xx} -\kappa^2u \  . \ee
 \begin{lemma}\label{2.1}
 	The three b.c (\ref{pnp}), together with the constraint $-\eps E_x=u$ can be introduced as a pair of non-local conditions:     
 	$$ u_x(0,t)=-\kappa^2\int_0^1 (1-s)u(s,t)ds-\eps  \kappa^2V, \ \ \   u_x(1,t)=\kappa^2 \int_0^1 su(s,t)ds-\eps\kappa^2V$$
 \end{lemma}
 \begin{proof}
 	By the field equation and the boundary condition (\ref{pnp}, \ref{bcc3}) admits a classical $C^2$ solution and $u(x,0)\geq 0$ 
 	$$ E(x)=E(0)+\int_0^xE^{'}(s)ds =\kappa^{-2}\eps^{-1}u_x(0)-\eps^{-1}\int_0^xu(s)ds $$
 	From $\int_0^1E=V$ we get
 	$$ V= \kappa^{-2}\eps^{-1}u_x(0) - \eps^{-1}\int_0^1\int_0^xu(s)dsdx= \eps^{-1}\left( \kappa^{-2}u_x(0) - \int_0^1(1-s) u(s)ds \right)  \ . $$
 Likewise
 	$$ E(x)=E(1)-\int_x^1E^{'}(s)ds =\kappa^{-2}\eps^{-1}u_x(1)+\eps^{-1}\int_x^1u(s)ds $$
and so
 $$ V= \eps^{-1}\kappa^{-2}u_x(1) + \eps^{-1}\int_0^1\int_x^1u(s)dsdx=  \eps^{-1}\left( \kappa^{-2}u_x(1) + \int_0^1 s u(s)ds \right)  \ . $$
 
 \end{proof}

\section{Properties of the  linrarized PNP}
We start from the following 
We start from the following 
\begin{prop}\label{prop1} 
	If equation (\ref{dhe}, \ref{bcc3}) admits a  classical solutions then
	\begin{description}
		\item{a.}
		The integral $\int_0^1u(x,t)dx$ is preserved.
		\item{b.} If $u(\cdot,0)$ is non-negative then  $u(\cdot ,t) $ is non-negative for any $t>0$. 
		\item{c.} For any $t>0$, $\|u(\cdot, t)\|_\infty \leq \cosh(\kappa/2) \|u(\cdot,0)\|_\infty$. 
	\end{description}
\end{prop}
\begin{proof} 
	(a): \ Follows immediately upon integration, taking advantage of the fact that the kernels of the integrals in (\ref{bcc3}) $(1-x)$ and $x$ sums to one. \\
	(b): \  Follows from an elementary observation involving the maximum principle. Indeed, let
	$u_\eps$ be a solution of the equation $u_{\eps,t}= u_{\eps, xx} +\kappa^2u+\eps$, under the boundary condition (\ref{bcc3}), where $\eps>0$.   Evidently $u_\eps\rightarrow u$ where $\eps\rightarrow 0$. 
	Let 
	$u(x,0)$ be strictly positive, and let $x_0\in[0,1]$, $t_0>0$ such that $u_\eps(x,t)>0$ for any $t\in[0, t_0)$, $x\in[0,1]$ and 
	$x\not= x_0$, $t=t_0$, while $u_\eps(x_0, t_0)=0$. From the boundary conditions we obtain that $u_{\eps, x}(0,t_0)<0$, $u_{\eps, x} (1,t_0)>0$, so $x_0\not= 0,1$. However, $u_{\eps, xx}(x_0, t_0)\geq 0$ since $x_0$ is an inner minimum. In particular $u_{\eps, t}\geq \eps$ by the equation. It follows that $u_\eps$ is, indeed strictly positive for any $\eps >0$, and the weak inequality is preserved in the limit $\eps =0$.\\
	(c):\ Assume, without limitation of generality, that $u(\cdot,0)\geq 0$. Using Lemma \ref{eigen}, or by a direct substitution, we obtain that $d^2\psi_0/d^2x -\kappa^2\psi_0=0$ where $\psi_0=\cosh(\kappa(x-1/2))$. Thus, $u(x,0)\leq \|u(\cdot,0)\|_\infty \psi_0(x)$ for any $x\in[0,1]$. In particular, $w(x,0):= \|u(\cdot,0)\|_\infty \psi_0(\cdot)-u(\cdot,0)$ is non-negative. Since $w$ satisfies 
	(\ref{dhe}, \ref{bcc3}) it follows by (b) that $w(\cdot,t)$ is non-negative for any $t>0$. Hence
	$0\leq u(x,t)\leq \|u(\cdot,0)\|_\infty \psi_0(\cdot)\leq  \|u(\cdot,0)\|_\infty \cosh(\kappa/2)$  for any $t\geq 0$. 
	
\end{proof} 
Let $u$ be a solution of (\ref{dhe}, \ref{bcc3}).  Substitute 
\be\label{utov} v(x,t)=e^{\kappa^2 t}u(x,t)\ . \ee
 Then  $v$ is a solution of
 \begin{multline}\label{pnpv} v_t=v_{xx} \ ,  x\in[(0,1), t>0 \\ v(\cdot, t)\in {\cal D}:= \left\{ w\in C^1[0,1]\cap C^2(0,1) ;   w_x(0)=-\kappa^2\int_0^1(1-s)w(s)ds, \ \ \ w_x(1)=\kappa^2\int_0^1sw(s)ds\right\}\ .  \end{multline}
\subsection{Eigenvalues and eigenfunctions}
The eigenfunctions of the operator $d^2/dx^2$ are given by $a\sin(\sa x)+b\cos(\sa x)$.  Substitute this in (\ref{pnpv}) we get 
\begin{multline}\label{eq10}a\left( \lambda^{1/2}+\kappa^2\lambda^{-1/2} -   \frac{\kappa^2}{\lambda}\sin(\lambda^{1/2})\right)+b\kappa^2\left( \frac{1-\cos\lambda^{1/2}}{\lambda}\right) =0\end{multline}
\begin{multline}\label{eq20}  a\left(   \lambda^{1/2}\cos(\lambda^{1/2})-\frac{\kappa^2 \sin(\lambda^{1/2})}{\lambda} +\frac{\kappa^2\cos(\lambda^{1/2})}{\lambda^{1/2}}
\right)\\
+b\left(
-\lambda^{1/2}\sin(\lambda^{1/2}) -\frac{\kappa^2\sin(\lambda^{1/2})}{\lambda^{1/2}} - \frac{\kappa^2\cos(\lambda^{1/2})-1}{\lambda}
\right)=0  \end{multline}
The system (\ref{eq10}, \ref{eq20}) is a linear system for the coefficients $a, b$.  
The determinant of this system is 
%$$ \frac{\kappa^2}{\lambda^{3/2}}(\cos(\lambda^{1/2})-1)   + \frac{2\kappa^2}{z}(\cos(\lambda^{1/2})-1) + \sin(\lambda^{1/2})(\kappa^2+\kappa^4\lambda^{-1}+\lambda)=
%$$
\begin{multline}\label{denum}Det(\lambda)=2\kappa^2 \lambda^{-1/2}(1-\cos(\sa))\left(1+ \frac{\kappa^2}{\lambda} \right)- \sin(\sa)\left(\frac{\kappa^2}{\sa}+\sa\right)^2 
\\
= \sin(\sa)\left(1+\frac{\kappa^2}{\lambda} \right)\left( \frac{2\kappa^2}{\sa} \tan (\sa/2) - \kappa^2-\lambda\right) \end{multline}
% \be\label{denum} 2\sin(\sqrt{\lambda}/2)(\lambda+1)\left[ 2\sin(\sqrt{\lambda}/2) -\lambda^{1/2}\cos(z /2)\right]   \ .  \ee
%\be\label{denum}Det(\lambda)=  -(z+\lambda^{-1/2})^2\sin(\lambda^{1/2}) - \frac{2(\cos(\lambda^{1/2})-1)}{z} \left(1+\frac{1}{\lambda}\right)   \ .  \ee
\begin{lemma}\label{eval}
	$\sa Det(\lambda )$ is a meromorphic function on the complex plane. 
	The roots of $Det(\lambda )=0$ are given by  $(2k\pi)^2$ where $k\in\mathbb{Z}$.  In addition
	$\lambda_m$, $m\in\mathbb{N}$ where $\{\lambda_m\}$ are the  roots of
	\be\label{tanz}2\tan(\lambda_m^{1/2}/2)=\lambda_m^{1/2}\kappa^{-2}(\lambda_m+\kappa^2) \ . \ee
	%All these roots are real and simple. 
	\par
	In addition, $\lambda=0$ is a root of $Det(\lambda )$ {\em only if}  $\kappa^2\not= 12$, and it coincides with a root $\lambda_1(\kappa)$ of (\ref{tanz}) as $\kappa^2\rightarrow 12$. 
	
	In addition, $\lambda=-\kappa^2$ is the only negative root of $Det$, and it is a simple one. 
\end{lemma}
%\begin{remark}\label{rem3.1}
%	In particular it follows that
%	$\lambda_m\in (4(m-1)^2\pi^2, (2m-1)^2\pi^2)$. 
%\end{remark}
\begin{proof}The case of non-zero roots 
	follows directly from (\ref{denum}). 
	
	To evaluate the case $\lambda=0$, let us rewrite the leading Taylor 
	expansion of the right side of (\ref{denum}) as a function of $\sa$. Using $\tan(\lambda^{1/2})=\sa+\lambda^{3/2}/3+2\lambda^{5/2}/15 \dots$, we expand (\ref{denum}) and obtain  that the leading terms in powers of $\lambda^{1/2}$ are
	$$ \sin(\lambda^{1/2})(1+\kappa^2/\lambda ) \left[\frac{2\kappa^2}{\sqrt{\lambda}}(\sqrt{\lambda}/2 + (\sqrt{\lambda}/2)^3/3+2(\sqrt{\lambda}/2)^5/15+\ldots )-\kappa^2-\lambda \right]= $$
	$$
	\sin(\lambda^{1/2})(1+\kappa^2/\lambda ) \left[(\frac{\kappa^2}{12}-1)\lambda +\kappa^2 \lambda^2/(120)+\ldots \right] \ . $$ 
\end{proof}
{\bf Conjecture:} {\it All rots of (\ref{tanz}) are real and simple.}

At this stage we can only show
that 
	there exists $R(\kappa)>0$ such that all roots of (\ref{tanz}) outside the disc $\{|z|<R(\kappa)\}$ are real and simple. Numerical test verifies, for all selected values of $\kappa$, that all roots inside the disc are real and simple as well. 

\begin{proof} Substitute $z=\lambda^{1/2}/2$ and set $f(z)=\tan(z)$, $g(z)=z\kappa^{1/2}(4z^2+\kappa^2)$ and $h(z)=g(z)-f(z)$.   Consider the orbit $\Gamma_n$  in the complex plan obtained by the edges of the square whose vertices are at $(n\pi, n\pi)$, $(-n\pi, n\pi)$, $(-n\pi, -n\pi)$, $(n\pi, -n\pi)$ where $n\in \mathbb{N}$ is large enough. The function $f$ is bounded uniformly along this orbit, while $|g|\rightarrow\infty$ uniformly on $\Gamma_n$ where $n\rightarrow \infty$. In particular, for $n$ large enough, $|g(z)|>|f(z)|$ for any $z\in \Gamma_n$. By the argument principle, $$\frac{1}{2\pi} \oint_{\Gamma_n}\frac{h^{'}}{h}dz = \frac{1}{2\pi}\oint_{\Gamma_n}\frac{g^{'}}{g}dz =\{number \ of\ zeroes \ of \ g\}$$
so
$$\frac{1}{2\pi} \oint_{\Gamma_n}\frac{h^{'}}{h}dz =\#\{zeroes\  of\  h\} -\#\{ poles \ of\  h\}=\#\{zeroes \ of\ g\}$$
in the interior of the square whose boundary is $\Gamma_n$.  Since $g$ is a polynomial  of order $3$, the number of zeroes of $g$ inside the square is $3$ for any $n$ large enough. Since all the poles of $h$ are identical to the poles of $f$, which  are  given by $(\pm k+1/2)\pi$ on the real line, $k\in \mathbb{N}\cup \{0\}$, there are exactly two poles of $h$ in $inter (\Gamma_{k+1})-inter(\Gamma_k)$ for all $k$ large enough, namely at $x=(\pm k+1/2)\pi$. Thus, there are exactly two zeroes (or one zero of order $2$) in $inter(\Gamma_{k+1})-inter(\Gamma_k)$. Evidently, there are two real roots in this domain, one in each interval $(k\pi, (k+1)\pi)$ and $(-(k+1)\pi, -k\pi)$. Thus these are the only roots in $inter (\Gamma_{k+1})-inter(\Gamma_k)$. It follows, in particular, that $h$ has only real roots outside a large enough square. 
\end{proof}
\begin{lemma}\label{eigen}
	The eigenvalues of the operator $d^2/dx^2$ under boundary conditions (\ref{pnpv}) are
	$\mu_{k}=(2k)^2\pi^2$, the roots $\lambda_m$ of (\ref{tanz}) and $\lambda_0=-\kappa^2$. 
	
The corresponding {\em unnormalized} eigenfunctions are:
\begin{itemize}
	\item $\mu_k=(2k\pi)^2 : \ \ \ \psi^{(1)}_k(x)=\cos(2k\pi x)$.  $k\in \mathbb{N}\cup \{0\}$.
	\item $\lambda_m \ \ (\ref{tanz}):  \ \ \ \psi^{(2)}_m(x)=\sin(\lambda_m^{1/2} (x-1/2))$, \   $m\in \mathbb{N}$.
	\item $\lambda_0=-\kappa^2$: \ \  $\psi_0(x)=\cosh(\kappa(x-1/2))$. 
	%$\psi_0(x)=\kappa\left(\cosh(\kappa)-1\right)\sinh(\kappa x)-\kappa\sinh(\kappa)\cosh(\kappa x)$. 
	\item If $\kappa^2=12$ then $\lambda_1=0$ and \\
	$\psi^{(2)}_1(x)=x-1/2= \lim_{\lambda\rightarrow 0} \lambda^{-1/2}\sin(\sa(x-1/2))$. 
\end{itemize}
\end{lemma}
\begin{proof}
The proof follows by  Lemma \ref{eval} and (\ref{eq10}, \ref{eq20}). If $\kappa^2\not=12$ then $0$ is {\em not} an eigenvalue, even though it is a root of $Det$. The reason is that the coefficients of (\ref{eq10}, \ref{eq20}) are degenerate in that case. However, if $\kappa^2=12$ then the first root $\lambda_1$ of (\ref{tanz}) is zero, and the eigenfunction follows by substitution. 
\end{proof}
%\begin{remark}
%Note that even though $4k\pi^2$ are roots of $Det$, while $(2k+1)^2\pi^2$ are not, it turns that the first a
%re {\em not} eigenvalues while the later are, actually, eigenvalues. See the section \ref{heatk} for an explanation. 
%\end{remark}
\subsection{The resolvent}
 The resolvent operator for Neumann problem on $[0,1]$ is expressed  in terms of the eigenvalues and eigenfunctions of the operator:
$$R_N(\lambda, x,y)=\frac{1}{\lambda}+ \frac{1}{2}\sum_{k=1}^\infty  \frac{\cos(k\pi x)\cos(k\pi y)}{\lambda- k^2\pi^2}$$
and
\begin{multline}\label{Nop}    \int_0^1 (1-s) R_N(\lambda,s,y) = \frac{1}{2\lambda}-A(\lambda, y)  \ ,  \\ 
\int_0^1 s R_N(\lambda,s,y) = \frac{1}{2\lambda}+A(\lambda, y)\  \end{multline}
where 
\be\label{Ay}A(\lambda, y):= \frac{1}{\pi^2} \sum_{m=1}^\infty \frac{2\cos((2m+1)\pi y)}{(2m+1)^2(\pi^2(2m+1)^2-\lambda)} \ . \ee
The resolvent $R$ corresponding to the boundary condition (\ref{pnpv}) can be written as
$$R(\lambda, x,y)= R_N(\lambda, x,y) + a(\lambda,y)\sin( \lambda^{1/2} x) +b(\lambda,y)\cos(\lambda^{1/2} x) $$
%$$:= \int_0^1\left[R_N(\lambda, x,y) + a(\lambda, y)\sin(z x)+b(\lambda,y)\cos(z x)\right] f(y)dy $$
From the boundary conditions of (\ref{pnpv}) and (\ref{Nop}) we obtain

\begin{multline}\label{eq1}a\left( \lambda^{1/2}+\kappa^2\lambda^{-1/2} -   \frac{\kappa^2}{\lambda}\sin(\lambda^{1/2})\right)+b\kappa^2\left( \frac{1-\cos\lambda^{1/2}}{\lambda}\right)+ \\ \kappa^2 \left[  \frac{1}{2\lambda}-A(\lambda, y)  \right]=0\end{multline}
\begin{multline}\label{eq2}  a\left(   \lambda^{1/2}\cos(\lambda^{1/2})-\frac{\kappa^2 \sin(\lambda^{1/2})}{\lambda} +\frac{\kappa^2\cos(\lambda^{1/2})}{\lambda^{1/2}}
\right)\\
+b\left(
-\lambda^{1/2}\sin(\lambda^{1/2}) -\frac{\kappa^2\sin(\lambda^{1/2})}{\lambda^{1/2}} - \frac{\kappa^2\cos(\lambda^{1/2})-1}{\lambda}
\right) \\
-\kappa^2\left[  \frac{1}{2\lambda}+A(\lambda, y)  \right]=0  \end{multline}

We can now solve (\ref{eq1},\ref{eq2})  for any $\lambda\not=0$ which is not a root of $Det$, 

\begin{multline}\label{detcom}a(\lambda,y)= \kappa^2 Det^{-1}(\lambda)\left\{\frac{1}{2\lambda}\left[\sin(\sa)\left(\sa +\frac{\kappa^2}{\sa}\right)-2\kappa^2\frac{1-\cos(\sa)}{\lambda}\right]\right. \\
\left. -A(\lambda,y) \sin(\sa)\left(\sa +\frac{\kappa^2}{\sa}\right)\right\}\\
b(\lambda,y)= \kappa^2 Det^{-1}(\lambda)\left\{ \frac{1}{2\lambda}\left[ \left(\sa +\frac{\kappa^2}{\sa}\right)(1+\cos(\sa))-\frac{2\kappa^2 \sin(\sa)}{\lambda}\right]\right. \\
\left. +A(\lambda,y) \left(\sa +\frac{\kappa^2}{\sa}\right)(1-\cos(\sa))\right\}
\end{multline}  
After some trigonometric manipulations on (\ref{detcom}, \ref{denum}) 
%and using the expression for $\psi^{(2)}_k$ from Lemma \ref{eigen}
we obtain
\begin{multline}\label{17} R(\lambda, x,y)-R_N(\lambda,x,y)= a(\lambda,y)\sin(\sa x) + b(\lambda, y)\cos(\sa x) =\\
 \frac{\kappa^2 \sin(\sa/2)A(\lambda, y) (\sa+\kappa^2/\sa)}{Det(\lambda)}
\sin\left(\frac{\sa(1-2x)}{2}\right) \\ 
- \frac{\kappa^2 \cos(\sa(x-1/2))}{2\sa\sin(\sa/2)(\lambda+\kappa^2)}
\end{multline}
\section{The heat kernel} \label{heatk}
To obtain the  heat kernel corresponding to the equation (\ref{pnpv}) we use (\ref{hk}) to obtain $K(x,y,t)=$
\be\label{Kerfull}\frac{1}{2\pi i}\lim_{T\rightarrow\infty}\int_{\gamma-iT}^{\gamma+iT} e^{-\lambda t}(R(\lambda,x,y)-R_N(\lambda,x,y))d\lambda +
\int_{\gamma-iT}^{\gamma+iT} e^{-\lambda t}R_N(\lambda,x,y)d\lambda\ee
We now recall that the second term above is just the heat kernel of the {\em Neumann} problem. This can be expanded in eigenfunctions:
\begin{multline}\label{KerN}K_N(x,y,t)= \frac{1}{2\pi i}\lim_{T\rightarrow\infty}\int_{\gamma-iT}^{\gamma+iT} e^{-\lambda t}R_N(\lambda,x,y)
d\lambda \\ = 1+\frac{1}{2}\sum_{k=1}^\infty  e^{-k^2\pi^2t}cos(k\pi x)\cos(k\pi y)\end{multline}

Then we calculate the residues  of $(R-R_N)e^{-\lambda t}$ using (\ref{17}). The residue at the pole $\lambda=-\kappa^2$ due to the second term in (\ref{17}) is 
\be\label{res0}\frac{\kappa\cosh(\kappa/2)\cosh(\kappa(x-1/2))}{2\sinh(\kappa/2)}e^{\kappa^2 t} \ee
%The other poles are at positive $\lambda$ (excluding the singular limit $\kappa^2=12$). 
%The coefficient of $\cos$ at (\ref{17}) has no poles, while the coefficient of $\sin$ has a pole due to the zero of $Det$ at $\lambda^{(2)}_k$, given by 
Let us now evaluate the other poles of   (\ref{17}). The poles of the first term (the coefficients of $\sin(\sa(x-1/2)$) are originated by two sources: Since $\sin(\sa/2)/Det(\lambda)$  has no singularity at $\lambda=(2m\pi)^2$, the only singularity due to $Det(\lambda)$ are the roots of (\ref{tanz}), i.e at $\lambda=\lambda_m$. The residue Theorem at this singularities yield
\be\label{resAm} \frac{\kappa^2 A(\lambda _m, y) (\sqrt{\lambda_m}+\kappa^2/\sqrt{\lambda_m})}{2\sqrt{\lambda_m}Det^{'}(\lambda_m)}
\sin\left(\frac{\sqrt{\lambda_m}(1-2x)}{2}\right) e^{-\lambda_m t}\ . \ee
where $A(\lambda,y)$ as given in (\ref{Ay}). However, the first term of (\ref{17}) contains also the poles at $\lambda=(2k+1)\pi$ due to the singularity of $A(\cdot,y)$ at these points. A direct calculation implies that the residue at these poles are precisely 
\be\label{resnoA}-\frac{1}{2}\cos((2k+1)x) \cos((2k+1)y)e^{-4k^2\pi^2 t} \  \ee
which eliminate the sum of odd indices in the Neumann heat kernel  (\ref{KerN}). 

The second term in (\ref{17}) also contain poles at $\lambda=(2k\pi)^2$, $k\in\mathbb{N}\cup \{0\}$. The sum of the resides  is 
\be\label{restoo}\kappa^2\sum_{k=1}^\infty \frac{\cos(2k\pi x) e^{-4k^2\pi^2 t}}{2 k\pi(4k^2\pi^2+\kappa^2)} -1 \ . \ee

%Note that $\phi(x,0)=0$ and $\lim_{t\rightarrow\infty}\phi(x,t)=-1$ uniformly in $x\in[0,1]$. 

Summarizing (\ref{res0}-\ref{restoo}) in (\ref{Kerfull}), using (\ref{KerN})     
and taking into account (\ref{utov}) we obtain 
\begin{multline} \label{hetk}K(x,y,t)=\frac{\kappa\cosh(\kappa/2)\cosh(\kappa(x-1/2))}{2\sinh(\kappa/2)}\\+\sum_{k=1}^\infty
\frac{\kappa^2 A(\lambda _k, y) (\sqrt{\lambda_k}+\kappa^2/\sqrt{\lambda_k})}{2\sqrt{\lambda_k}Det^{'}(\lambda_k)}
\sin\left(\frac{\sqrt{\lambda_k}(1-2x)}{2}\right) e^{-(\lambda_k+\kappa^2) t} \\
+ \frac{1}{2}\sum_{k=1}^\infty  \cos(2k\pi x)
\left( \cos(2k\pi y)+\frac{1}{ k\pi(4k^2\pi^2+\kappa^2)} 
\right)
e^{-(4k^2\pi^2 +\kappa^2)t}\end{multline}
\begin{cor}\label{4.1}
	The real eigenfunctions of the {\em adjoint }  operator ${\cal L}_0^\dag$ under boundary conditions ${\cal D}$  (\ref{pnpv})  are:
	\begin{itemize}
		\item $\mu_k: \ \ \ \phi^{(1)}_k(y)=\cos(2k\pi y)+\frac{1}{ k\pi(4k^2\pi^2+\kappa^2)} $.  $k\in \mathbb{N}$.
		\item $\lambda_k: \ \ \ \phi^{(2)}_k(y)=A(\lambda_k,y)$, \   $k\in \mathbb{N}$.
		\item $\lambda_0=-\kappa^2$: \ \ $\phi^{(0)}(y)=1$. 
	
	\end{itemize}
\end{cor}
\section{Conclusions}\label{conc}
The leading term in the eigenfunctions expansions of the PNP equation (\ref{dhe}, \ref{bcc3}) is the stationary term proportional to $\cosh(\kappa(x-1/2))$. 
The other modes decay exponentially. Out of these decaying mode, the leading one decays as $\exp(-(\lambda_1+\kappa^2)t)$ where $\lambda_1\in(0, \pi^2)$. In general, the decaying modes correspond to two sets: $\mu_n$ decay as $\exp(-(4n^2\pi^2+\kappa^2)t)$, and $\lambda_n$ decay as $\exp(-(\lambda_n+\kappa^2)t)$ where $\lambda_n\in (4(n-1)^2\pi^2, (2n-1)^2\pi^2)$ . 

The heat kernel associated with the operator ${\cal L}_0:= d^2/dx^2$  on the domain ${\cal D}$ satisfying the boundary conditions (\ref{bcc3}) is given by (\ref{calL0}). The operator ${\cal L}_0$ itself, acting on a function $h\in {\cal D}$, takes the form of a formal series of the eigenstates $\psi_n^{(1,2)}$ and $\psi_0$:
\begin{multline} {\cal L}_0 h(x)= \frac{\partial}{\partial t} K|_{t=0}*h=\\
\frac{\kappa^3\cosh(\kappa/2)}{2\sinh(\kappa/2)}<\phi_0,h>\psi_0(x) \\ -\sum_{n=1}^\infty\lambda_n
\frac{\kappa^2 (\sqrt{\lambda_n}+\kappa^2/\sqrt{\lambda_n})}{2\sqrt{\lambda_n}Det^{'}(\lambda_n)}
<\phi_n^{(1)},h>\psi^{(1)}_n(x)\\
-2n^2\pi^2 \sum_{n=1}^\infty <\phi_n^{(2)},h>\psi_n^{(2)}(x)
\end{multline}
where $<\phi,h>:= \int_0^1 h(y)\phi(y)dy$. 
An interesting conclusion concerns  the adjoint of the operator ${\cal L}_0^\dag$.  Its heat kernel $K^\dag$ is obtained by swapping $x$ and $y$ in $K$, namely $K^\dag(x,y,t)=K(y,x,t)$. Thus
\begin{multline} {\cal L}_0^\dag h(x)=
\frac{\kappa^3\cosh(\kappa/2)}{2\sinh(\kappa/2)}<\psi_0,h>\phi_0(x) \\ -\sum_{n=1}^\infty\lambda_n
\frac{\kappa^2 (\sqrt{\lambda_n}+\kappa^2/\sqrt{\lambda_n})}{2\sqrt{\lambda_n}Det^{'}(\lambda_n)}
<\psi_n^{(1)},h>\phi^{(1)}_n(x)\\
-2n^2\pi^2 \sum_{n=1}^\infty <\psi_n^{(2)},h>\phi_n^{(2)}(x)
\end{multline}

The associated eigenfunctions, given by $ \phi^{(1)}_k(y)=\cos(2k\pi y)+\frac{1}{ k\pi(4k^2\pi^2+\kappa^2)} $ and $\phi^{(2)}_k(y)=A(\lambda_k,y)$ are not trigonometric functions. In  particular we {\em cannot} identify 
${\cal L}_0^\dag$ with $d^2/dx^2$ on a certain domain ${\cal D}^\dag$, as we did for ${\cal L}_0$. 

\vskip .3in 
\noindent
{\bf Open question:} \ {\it Find an explicit expression  for  the generator to the adjoint operator ${\cal L}_0^\dag$ and its domain  ${\cal D}^\dag$. 
	
	\vskip .5in 
	
	{\bf Acknowledgment} : This research was supported by the ISF research  grant 296/20. 
\appendix
\section{Appendix}

\subsection{From resulvent to heat kernel}
  From the resolvent to the  heat kernel
 Let $U(\lambda,x)$ be a solution of 
 \be\label{res} U_{xx} +\lambda U+f=0 \    \ee
 satisfying a well posed  boundary conditions,
where $\lambda\in \mathbb{C}$.  
Then $$U(\lambda, x)=\int_0^1 R(\lambda, x,y)f(y) dy $$
where $R$ is the Resolvent:
$$ \frac{\partial^2 R}{\partial x^2} + \lambda R + \delta_{x-y}=0$$

Suppose $U$ is analytic, as function of $\lambda$, in the half plane $Re(\lambda)\leq \gamma$ for some $\gamma\in\R$. Then
 \be\label{hk} u(x,t)=\frac{1}{2\pi i}\lim_{T\rightarrow\infty}\int_{\gamma-iT}^{\gamma+iT} e^{-\lambda t}U(\lambda,x)d\lambda \ \ee
 is the solution of (\ref{pnpv}). Indeed
 $$u_t=-\frac{1}{2\pi i}\lim_{T\rightarrow\infty}\int_{\gamma-iT}^{\gamma+iT} \lambda e^{-\lambda t}U(\lambda,x)d\lambda $$
 while, by (\ref{res}), 
 $$ u_{xx}= \frac{1}{2\pi i}\lim_{T\rightarrow\infty}\int_{\gamma-iT}^{\gamma+iT} e^{-\lambda t}U_{xx}(\lambda,x)d\lambda =\frac{1}{2\pi i} \lim_{T\rightarrow\infty}\int_{\gamma-iT}^{\gamma+iT} e^{-\lambda t}[-f-\lambda U(\lambda,x)]d\lambda  $$
 $$= u_t-f(x)\frac{1}{2\pi i} \lim_{T\rightarrow\infty}\int_{\gamma-iT}^{\gamma+iT} e^{-\lambda t}d\lambda  \ , $$
 while
 $$\frac{1}{2\pi i} \lim_{T\rightarrow\infty}\int_{\gamma-iT}^{\gamma+iT} e^{-\lambda t}d\lambda =\frac{1}{\pi }e^{-\gamma t}\lim_{T\rightarrow\infty}  \frac{\sin(tT)}{t}=e^{-\gamma t} \delta_{t=0}=\delta_{t=0}$$
 as a distribution.

The heat kernel can, then, be written as 
\be\label{Kernel} K(x,y,t)= \frac{1}{2\pi i}\lim_{T\rightarrow\infty}\int_{\gamma-iT}^{\gamma+iT} e^{-\lambda t}R(\lambda,x,y)d\lambda 
\ee
where $t\geq 0$ and $x,y\in [0,1]^2$.


\begin{thebibliography}{9}
 	\bibitem{St}	Stikonas, A: {\em The Sturm-Liouville problem with a nonlocal boundary condition}. Lith Math J 47, 336-351 (2007).
 	\bibitem{BS} Simo, B:  {\em
 		Advanced complex analysis. Part 2B : a comprehensive course in analysis}, Springer (2015)
 \bibitem{DS}
 Dunford, N and Schwartz, J.T: {\em  Linear Operators, part I}, Wiley Classics Library, (1988)
 	\bibitem{Feller1}
 	Feller,W: {\em The Parabolic Differential Equations and the Associated Semi-Groups of
 		Transformations}, Annals of Mathematics , May,  Second Series, Vol. 55, pp.
 	468-519, (1952)
 	
 		\bibitem{Feller2}
 		Feller,W:{\em
 		Diffusion processes in one dimension}, 
 		Trans. Amer. Math. Soc., 77 , pp. 1-31, (1954)
 		\bibitem{Ardent}
 Arendt, W, Kunkel,S and  Kunze, M:
{\em
 Diffusion with nonlocal boundary conditions}, 
 Journal of Functional Analysis,
 Volume 270, Issue 7,
 Pages 2483-2507,  (2016)
 	
 	\bibitem{debay} Debye,P  and  Falkenhagen,H : Phys. Z.  29, 121
 \bibitem{Bo} Kopytko,B and  Shevchuk,R :
 {\em Diffusion in one-dimensional bounded domains with reflection, absorption and jumps at the boundary and at some interior point},
 		Journal of Applied Mathematics and Computational Mechanics,
 	12,
 55-68, (2013)
 \bibitem{kol} Kolmogoroff, A. {\em About the analytical methods in probability theory},  Math. Ann. 104 : 415-458 (1931).
 \bibitem{BM}  Karatzas, I and Shreve, S.E , {\em Brownian Motion and Stochastic Calculus}, Graduate Text in Mathematics, Springer1991
 \bibitem{pnp} Golovnev, A , and Trimper, S: {\em
Exact solution of the Poisson-Nernst-Planck
equations in the linear regime}, 
 J. Chem. Phys. 131, 114903 (2009)
 \bibitem{mem}  Pabst, M, 
{\em  Analytical solution of the Poisson-NernstPlanck equations for an electrochemical
 system close to electroneutrality}, 
  J. Chem. Phys. 140, 224113 (2014)

\bibitem{ros1}
Ben-Ari, I and Pinsky, R:  {\em  Spectral analysis of a family of second-order elliptic operators with nonlocal boundary condition indexed by a probability measure}. J. Funct. Anal. 251 (2007), no. 1, 122-140.
 \end{thebibliography}
\end{document}